\documentclass{IEEEtran4PSCC}

\usepackage{cite}

\ifCLASSINFOpdf
   \usepackage[pdftex]{graphicx}
\else
   \usepackage[dvips]{graphicx}
\fi
\usepackage[cmex10]{amsmath}
\usepackage{algorithmic}

\usepackage{array}
\makeatletter
\let\old@ps@headings\ps@headings
\let\old@ps@IEEEtitlepagestyle\ps@IEEEtitlepagestyle
\def\psccfooter#1{%
    \def\ps@headings{%
        \old@ps@headings%
        \def\@oddfoot{\strut\hfill#1\hfill\strut}%
        \def\@evenfoot{\strut\hfill#1\hfill\strut}%
    }%
    \def\ps@IEEEtitlepagestyle{%
        \old@ps@IEEEtitlepagestyle%
        \def\@oddfoot{\strut\hfill#1\hfill\strut}%
        \def\@evenfoot{\strut\hfill#1\hfill\strut}%
    }%
    \ps@headings%
}
\makeatother

\usepackage{bbm}
\usepackage{algorithm}		
\usepackage{algorithmic}		
\usepackage{amssymb}
\usepackage{tikz}
\usetikzlibrary{shapes,arrows,calc}
\DeclareMathAlphabet{\mathdutchcal}{U}{dutchcal}{m}{n}
\newtheorem{rem}{Remark}
\newtheorem{defi}{Definition}
\newtheorem{prop}{Proposition}

\newtheorem{ass}{Assumption}

\newenvironment{proof}{\vspace{0cm}\paragraph*{Proof}}{\vspace{0cm}\hfill$\blacksquare$}
\usepackage{enumitem} 
\usepackage{changes}

\usepackage{accents}

\usepackage[hyperindex=true, %
  pdftitle={Adaptive Real-Time Grid Operation via Online Feedback Optimization with Sensitivity Estimation},%
  pdfauthor={Miguel Picallo, Lukas Ortmann, Saverio Bolognani, Florian Dörfler},%
  colorlinks=true,%
  pagebackref=false,%
  plainpages=false,%
  pdfpagelabels,%
  linkcolor=black,%
  citecolor=black,%
  filecolor=black,%
  urlcolor=black%
]{hyperref}

\renewcommand{\r}{\textcolor{black}}
\newcommand{\rr}{\textcolor{black}}
\newcommand{\normsz}[1]{\lVert #1 \rVert}

\begin{document}
%
\title{Adaptive Real-Time Grid Operation via Online Feedback Optimization with Sensitivity Estimation}

\author{
\IEEEauthorblockN{Miguel Picallo\IEEEauthorrefmark{1}, Lukas Ortmann\IEEEauthorrefmark{1}, Saverio Bolognani, Florian D{\"o}rfler}
\IEEEauthorblockA{Automatic Control Laboratory, ETH Zurich, 8092 Zurich, Switzerland \\ \{miguelp,ortmannl,bsaverio,dorfler\}@ethz.ch}
}

\maketitle

\begin{abstract}
In this paper we propose an approach based on an Online Feedback Optimization (OFO) controller with grid input-output sensitivity estimation for real-time grid operation, e.g., at subsecond time scales. The OFO controller uses grid measurements as feedback to update the value of the controllable elements in the grid, and track the solution of a time-varying AC Optimal Power Flow (AC-OPF). Instead of relying on a full grid model, e.g., grid admittance matrix, OFO only requires the steady-state sensitivity relating a change in the controllable inputs, e.g., power injections set-points, to a change in the measured outputs, e.g., voltage magnitudes. Since an inaccurate sensitivity may lead to a model-mismatch and jeopardize the performance, we propose a recursive least-squares estimation that enables OFO to learn the sensitivity from measurements during real-time operation, turning OFO into a model-free approach. We analytically certify the convergence of the proposed OFO with sensitivity estimation, and validate its performance on a simulation using the IEEE 123-bus test feeder, and comparing it against a state-of-the-art OFO with constant sensitivity.
\end{abstract}


\thanksto{\noindent \IEEEauthorrefmark{1} These two authors contributed equally.
\newline
Funding by the Swiss Federal Office of Energy through the projects “ReMaP” (SI/501810-01) and “UNICORN” (SI/501708), the Swiss National Science Foundation through the NCCR Automation, and by the ETH Foundation is gratefully
acknowledged.}

\section{Introduction}
The increasing amount of controllable, yet sometimes unpredictable, power resources in electrical grids, e.g., renewable generation, electric vehicles, flexible loads, etc., leads to new challenges and opportunities in the operation of power systems. On the one hand, these new controllable elements allow to minimize the grid operational cost and promote a transition to a more sustainable power system. On the other hand, given the volatility and unpredictability of these resources, fast control decisions are required to avoid constraint violations, e.g., overvoltages. This is especially relevant in distribution grids, where many of these resources are deployed. However, measurement scarcity and poor grid models challenge grid operation at such low voltage levels.

One way to leverage the controllability of these resources and to optimize the grid operation is by solving an AC Optimal Power Flow (AC-OPF) \cite{molzahn2019surveyrel}, an optimization problem to determine the set-points of controllable resources that minimize the operational cost and enforce grid safety requirements, e.g., voltage limits, line thermal limits, etc. Unfortunately, standard AC-OPF requires a) full grid observability, e.g., measurements of all active and reactive power injections and consumptions, and b) an accurate nonlinear grid model, e.g., its admittance matrix \cite{molzahn2019surveyrel}. Yet, learning the model may require an extensive deployment of measurements across the network \cite{bolognani2013topol, moffat2020pmuimpedance}, usually not available or affordable on the distribution system level. Furthermore, the volatility of renewable energy sources and household loads requires high sampling and control-loop rates to satisfy the grid constraints. Yet, solving a computationally expensive AC-OPF may pose a limit on these rates.

Online Feedback Optimization (OFO) \cite{molzahn2017survey, hauswirth2017online, anese2016optimal} is a novel computationally efficient approach that allows to track the solutions of an AC-OPF problem under time-varying conditions using subsecond control-loop rates. OFO is based on a controller that uses grid measurements as feedback to iteratively steer the controllable input set-points towards the AC-OPF solutions, and has already been successfully tested in both simulations and experimental settings \cite{ortmann2020experimental}. Furthermore, OFO neither requires full grid observability \cite{picallo2020seopfrt}, nor an accurate nonlinear grid model. It only needs measurements of the outputs that need to be controlled, and the input-output sensitivity that matches a change in the input to a change in the output. \r{This sensitivity is essentially a derivative of the power flow equations at the operating point \cite{bolognani2015fast}, and thus depends on the grid state and exogenous disturbances, e.g., loads. Hence, constructing an accurate sensitivity requires the grid model and full measurements of the grid to evaluate it. To avoid these requirements, some OFO approaches use a constant approximate linear model, and thus a constant approximate sensitivity \cite{anese2016optimal,picallo2020seopfrt,ortmann2020experimental}.} Even though OFO is robust against small approximation errors in this sensitivity \cite{ortmann2020experimental}, an inaccurate sensitivity introduces a model-mismatch that may lower the approach performance \cite{colombino2019towards}.
\rr{Therefore, some model-free approaches try to operate the system optimally without requiring a model or sensitivity. First, reinforcement learning allows to disregard the model, and instead take decisions based solely on measurements \cite{chen2021reinforcement}. However, reinforcement learning has limited theoretical guarantees, and may not be able to enforce the grid safety constraints during its learning phase. Second, data-driven control \cite{coulson2019data, mugnier2016model, bianchin2021data} based on Willems Fundamental lemma \cite{willems2005note} allows to compute the sensitivity after gathering sufficient data. Yet, these approaches estimate a constant linear model, and thus may fail to adapt to different operating points. Finally, zeroth-order gradient-free methods as \cite{chen2021model} allow to operate the system while continuously estimating and updating the sensitivity. However, \cite{chen2021model} requires a sufficient time-scale separation between the sensitivity estimation procedure and the feedback optimization, which may lower the convergence rate of the entire approach if the measurement sample rate is restricted due to communication limits.}

Therefore, in this paper, \rr{with a similar spirit as in the extremum seeking approach \cite{chen2021model}}, we propose a model-free OFO approach that sequentially estimates \rr{a time-varying} sensitivity while operating the grid, bypassing the need to know the whole grid model accurately, and to have full grid observability.
Our contributions are as follows: First, we design a sensitivity learning approach via recursive least squares \cite{lennart1999system,isermann2010identification}. We use as measurements the change in the outputs caused by a change of the controllable inputs. Second, we combine this sensitivity estimation with a persistently exciting OFO that gathers enough information about the sensitivity while driving the control inputs towards the AC-OPF solutions. Third, we certify the convergence of both the estimated sensitivity and the control input towards the true sensitivity and the time-varying solution of the AC-OPF, respectively. Fourth and finally, we simulate the proposed OFO controller with sensitivity estimation on the 3-phase, unbalanced IEEE 123-bus test feeder \cite{kersting1991radial} using real consumption data, and show its superior performance over a state-of-the-art OFO with a constant sensitivity approximation.

The paper is structured as follows: Section~\ref{sec:pre} presents some preliminaries on grid models, AC-OPF and OFO. Section~\ref{sec:ofosensi} explains our proposed OFO with sensitivity estimation approach, and provides theoretical convergence guarantees. Section~\ref{sec:test} shows the simulation on a test feeder. Finally, Section~\ref{sec:conc} concludes and discusses further work.

    

\section{Preliminaries: Grid Model, AC-OPF and OFO}\label{sec:pre}

\subsection{Grid Model}
For each bus $i$ of a $n$-bus power system we define the voltage magnitude as $v_i \in \mathbb{R}$, the active and reactive power as $p_i \in \mathbb{R}$ and $q_i \in \mathbb{R}$, respectively. We obtain the vectors $v$, $p$, and $q$ of dimension $n$ by stacking the individual bus quantities, i.e., $v=[v_1,\dots,v_n]^T$. We define the control input vector $u \in \mathbb{R}^{n_u}$ consisting of all the controllable resources (e.g. active and reactive generation and flexible loads in $p$ and $q$, slack bus voltage magnitude $v_1$ through tap changers); the output vector $y$ (e.g. voltage magnitude elements in $v$) with all the quantities that we measure and want to control through the inputs; and the disturbance vector $d$ with all uncontrollable power injections (e.g. conventional consumption loads in $p$ and $q$). The grid admittance matrix and the power flow equations allow to define an input-output map $\mathdutchcal{h}(\cdot)$ that characterizes the output $y$ as a non-linear function of $u$ and $d$:
\begin{equation}\label{eq:ioPF}
    y= \mathdutchcal{h}(u,d).
\end{equation}
The input-output map $\mathdutchcal{h}(\cdot)$ is not typically available in closed form, since in general it is not possible to derive an analytical expression of $v$ (in $y$) as a function of $p$ and $q$ (in $u$ and $d$) using the power flow equations \cite{molzahn2019surveyrel}. Yet, the local existence of a continuous differentiable map $\mathdutchcal{h}(\cdot)$ can be guaranteed by the implicit function theorem \cite{krantz2012implicit}. 


\subsection{AC Optimal Power Flow for Grid Operation}
The operation of a power grid consists of deciding the input $u_t$ at each time instant $t$. An AC-OPF allows to formulate this decision process as an optimization problem:
\begin{align}\label{eq:OPF}
        \begin{split}
            u^*_t,y_t^* = & \arg\min_{u \in \mathcal{U}_t,y} f(u) + g(y) \\
            & \text{ s.t. } y=\mathdutchcal{h}(u,d_t),
        \end{split}
    \end{align}
where $f(u)$ is the operational cost on the input $u$; $g(y)$ is a penalty function to enforce some grid specification on the output $y$, e.g., voltage limits; $\mathcal{U}_t$ is the time-varying set of admissible inputs that defines the operational constraints on $u_t$, e.g., power limits $\mathcal{U}_t = \{u\,|\,\underline{u}_t<u<\overline{u}_t\}$; and $d_t$ is the disturbance value at time $t$, e.g., uncontrollable loads or non-dispatchable generation. The nonlinear input-output model \eqref{eq:ioPF} in \eqref{eq:OPF} relates the outputs to the chosen input.



Optimal real-time decision making consists of first taking measurements $d_{t}$; then, solving the AC-OPF problem \eqref{eq:OPF}, and finally applying the solution $u^*_t$ to the system. Then, this is repeated at the next time step $t+1$.

\subsection{Linear Power Flow Approximation}

Solving AC-OPF problems \eqref{eq:OPF} to determine the set-points of power resources is a compelling and valuable tool for grid operators, but it comes with some drawbacks: First, the full nonlinear model of the grid $\mathdutchcal{h}(u,d)$ is needed. Second, solving the AC-OPF~\eqref{eq:OPF} can be computationally expensive, which may jeopardize its use for real-time grid operation. This can be circumvented by linearizing the map $\mathdutchcal{h}(\cdot)$ in \eqref{eq:ioPF} at an operating point \cite{low2014convex, bolognani2015existence, molzahn2019surveyrel}, e.g., the zero-injection point $(u_\text{op},d_\text{op})=(0,0)$, to obtain the approximation
\begin{align}\label{eq:linPF}
    \begin{split}
        y = H_0 u + D_0 d + y_0,
    \end{split}
\end{align}
where $y_0$ is an offset representing the output value when $u=d=0$, e.g., $1$~p.u. for all voltage magnitudes. The matrices $H_0=\nabla_u \mathdutchcal{h}(u,d)|_{(u_\text{op},d_\text{op})}$ and $D_0=\nabla_d \mathdutchcal{h}(u,d)|_{(u_\text{op},d_\text{op})}$ are evaluated at the operating point, and represent the sensitivities of the output with respect to changes in the input $u$ and disturbance $d$, respectively. This linear approximation \eqref{eq:linPF} can substitute the nonlinear map $\mathdutchcal{h}(\cdot)$ in the AC-OPF \eqref{eq:OPF} to get
\begin{align}\label{eq:linOPF}
        \begin{split}
            \min_{u \in \mathcal{U}_t} f(u) + g(H_0 u + D_0 d_t + y_0).
        \end{split}
\end{align}

\subsection{Online Feedback Optimization (OFO)}


Solving the AC-OPF with linear power flow approximation \eqref{eq:linOPF} is computationally efficient and could be employed in real-time operation. However, this approach does not take advantage of output measurements $y_t$, since it only feeds $d_t$ through the inaccurate linear model \eqref{eq:linPF}. Hence, such a \textit{feedforward} approach introduces a model-mismatch that can cause a performance degradation, and even lead to constraint violations, e.g., under and overvoltages.

Instead, OFO is a novel approach \cite{hauswirth2017online,anese2016optimal,molzahn2017survey} that uses $y_t$ as \textit{feedback} to achieve a safer grid operation and track the solution of the AC-OPF \eqref{eq:OPF} under time-varying conditions. 
For that, OFO turns a standard optimization algorithm, in our case projected gradient decent  \cite{bertsekas1997nonlinear}, into a feedback controller that takes the grid output measurements $y_t$, instead of computing the output $y_t$ via the grid model \eqref{eq:ioPF} or the linearized one \eqref{eq:linPF}. Projected gradient decent  consists of a gradient step and a projection: First, we compute the gradient of the cost function in \eqref{eq:linOPF}:
\begin{align*}
    \nabla_u \big(f(u)+g(y)\big)\overset{\eqref{eq:linPF}}{=}\nabla_u f(u) + H_0^T \nabla_y g(y).
\end{align*}

To minimize the operational cost, the current input $u_t$ is pushed along the direction of the negative gradient with a step size $\alpha$, and then it is projected onto the feasible space $\mathcal{U}_t$ to enforce the operational constraints on the input, i.e.,
\begin{align}\label{eq:ofo}
    \begin{split}
        u_{t+1} = \Pi_{\mathcal{U}_t}\big[u_t - \alpha \big(\nabla_u f(u_t) + H_0^T \nabla_y g(y_t)\big)\big],
    \end{split}
\end{align}
where $\Pi_{\mathcal{U}}\big[u]=\arg \min_{z \in \mathcal{U}} \normsz{u-z}_2^2$ is the projection of $u$ onto $\mathcal{U}$, which is typically easy to evaluate for power grid operation \cite{anese2016optimal}, especially if $\mathcal{U}_t = \{u\,|\,\underline{u}_t \leq u \leq \overline{u}_t\}$ is a box constraint.

\section{Online Feedback Optimization with Sensitivity Estimation} \label{sec:ofosensi}

The OFO controllers are robust, \r{i.e., preserve stability,} against using a constant power flow sensitivity approximation $H_0$ instead of the actual one $\nabla_u \mathdutchcal{h}(u,d)$  \cite{ortmann2020experimental,colombino2019towards}. Unfortunately, even if the overall system is stable, a model mismatch between $H_0$ and $\nabla_u \mathdutchcal{h}(u,d)$ may lead to a difference between the solution $u_t^*$ of the AC-OPF problem and the values $u_t$ produced by the OFO controller \eqref{eq:ofo} \cite{colombino2019towards}.
Therefore, we propose an approach to sequentially update the sensitivity $H_0$ into a good approximation of the true sensitivity $\nabla_u \mathdutchcal{h}(u,d)$, and thus avoid a potential performance degradation. For that, we will consider the sensitivity as a time-varying parameter $H_t=\nabla_u \mathdutchcal{h}(u_t,d_t)$, and propose a recursive least-squares approach to generate sensitivity estimates $\hat{H}_t$ using the measured variations of $y$ and $u$ over time, $\Delta u$ and $\Delta y$ respectively.
Then, in every time step we feed this estimated sensitivity $\hat{H}_t$ to the OFO as in Figure~\ref{fig:block_diagram}. 

\tikzstyle{block} = [draw, fill=white, rectangle, 
    minimum height=2em, minimum width=3em]
\tikzstyle{block2} = [draw, fill=white, rectangle, 
    minimum height=1.5em, minimum width=3em]
\tikzstyle{sum} = [draw, fill=white, circle, node distance=1cm]
\tikzstyle{input} = [coordinate]
\tikzstyle{output} = [coordinate]
\tikzstyle{pinstyle} = [pin edge={to-,thin,black}]

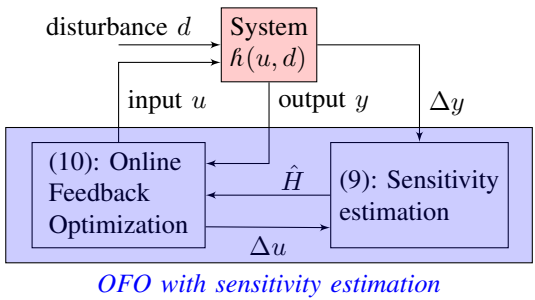
\begin{figure}[b]
\centering
\begin{tikzpicture}[auto, node distance=2cm,>=latex']
	\node [input] (input0) {};   
    \node [input, right of =input0, node distance = 1cm,align=left] (sat) {};
    \node [block, right of = sat, node distance = 2cm, align=left, fill = red, fill opacity=0.2, text opacity = 1] (sys) {System \\ $\mathdutchcal{h}(u,d)$};
    \node [input, right of = sys, node distance = 2cm,align=left, minimum height = 1.cm] (meas) {};
    
    \node [block, below of = sat, node distance = 2cm,align=left, minimum height = 1.4cm, minimum width = 2.3cm] (opt) {\eqref{eq:ofope}: Online \\ Feedback \\ Optimization};
    \node [block, below of = meas, node distance = 2cm,align=left, minimum height = 1.4cm] (est) {\eqref{eq:sensi}: Sensitivity \\ estimation};
    
    \draw[->] (est) -- node[name=h,above=-1.5pt,pos=0.3] {$\hat{H}$} (opt);
    \draw[->] (opt.340) -- node[name=du,below] {$\Delta u$} (est.200);
    
    \draw[->] (sys) -| node [name=y,right,pos=0.8] {$\Delta y$} (est);
    \draw[->] (opt) |- node [name=u,right,pos=0.25] {input $u$} (sys.200);
    \draw[->] (sys) |- node [name=x,right,pos=0.125] {output $y$} (opt.20);

    \node [block, below of =sys, node distance = 2cm,fill = blue, fill opacity=0.2, minimum height = 1.8cm, minimum width = 7cm] (algblock) {};
    \node[text=blue, below of = sys, node distance =3.2cm] (optest) {\textit{OFO with sensitivity estimation}};
        
    \node [input,left of = sys, node distance = 2cm] (inpdist) {};
    \draw[->] (inpdist) -- node[name=dist,above,pos =0.] {disturbance $d$} (sys);    
\end{tikzpicture}

\caption{Model-free grid operation via Online Feedback Optimization (OFO) with sensitivity estimation.}\textbf{\label{fig:block_diagram}}
\end{figure}

\subsection{Sensitivity Estimation}

\r{Due to the non-linearity of $\mathdutchcal{h}(u,d)$, the true sensitivity $\nabla_u \mathdutchcal{h}(u,d)$ depends on the values of $u$ and $d$. The temporal variation of the disturbance $d_t$ and the input $u_t$, e.g., due to applying the OFO controller \eqref{eq:ofo} in the input case, produces a time-varying sensitivity $H_t=\nabla_u \mathdutchcal{h}(u_t,d_t)$.} Instead of learning the dependency on $u$ and $d$, we model a time-varying sensitivity $H_t$ with the following random process:
    \begin{align}\label{eq:Hdyn}
        \begin{split}
            h_{t} = h_{t-1} + \omega_{p,t-1}
        \end{split}
    \end{align}
where $h=\text{vec}(H)$ is the column-wise vector representation of the sensitivity matrix $H$, $\Delta u_{t-1} = u_{t}-u_{t-1}$ denotes a change of the input $u$, and $\omega_{p,t} \sim \mathcal{N}(0,\Sigma_{p,t})$ is a Gaussian process noise with covariance $\Sigma_{p,t}=\Sigma_{p_1}+\Sigma_{p_2}\normsz{\Delta u_{t}}_2^2$, that represents how the sensitivity changes over time. 
\r{We make the part $\Sigma_{p_2}$ of the process noise proportional to $\normsz{\Delta u_t}_2$, since a large $\Delta u_t$ can trigger a larger change in the true sensitivity $\nabla_u \mathdutchcal{h}(u,d)$ that depends on $u$, and the part $\Sigma_{p_1}$ independent of $\Delta u_t$ to account for a uncontrolled random change $\Delta d_{t}= d_{t+1}-d_t$ that can affect the sensitivity as well.}

Next, to derive a measurement equation for the sensitivity $H_t$,
\r{consider the first-order Taylor approximation of $y_t$
\begin{align}\label{eq:taylormeaseq}
        \begin{split}
           \overbrace{\mathdutchcal{h}(u_t,d_t)}^{y_{t}} \approx & \overbrace{\mathdutchcal{h}(u_{t-1},d_{t-1})}^{y_{t-1}} +  \overbrace{\nabla_u \mathdutchcal{h}(u_{t-1},d_{t-1})}^{H_{t-1}} \Delta u_{t-1} \\
           & + 
           \nabla_d \mathdutchcal{h}(u_{t-1},d_{t-1}) \Delta d_{t-1}.
        \end{split}
\end{align}    
At each time $t$,} we measure $y_{t}$, and compute the variation $\Delta y_{t-1}= y_{t}-y_{t-1}$. \r{Based on the Taylor approximation \eqref{eq:taylormeaseq},} we treat this variation $\Delta y_{t-1}$ as a noisy linear measurement of $H_{t-1}$ through a measurement model that depends on $\Delta u_{t-1}$: 
     \begin{align}\label{eq:measeq}
        \begin{split}
            \Delta y_{t-1} & = \underbrace{H_{t-1} \Delta u_{t-1}}_{=U_{\Delta,t-1} h_{t-1}} + \omega_{m,t-1} 
        \end{split}
    \end{align}   
where $U_{\Delta,t}=\Delta u_t^T \otimes \mathbbm{1}$, with the Kronecker product $\otimes$, and $\omega_{m,t} \sim \mathcal{N}(0,\Sigma_{m,t})$ is a Gaussian measurement noise with covariance $\Sigma_{m,t}=\Sigma_{m_1}+\Sigma_{m_2}\normsz{\Delta u_{t}}_2^2+\Sigma_{m_3}\normsz{\Delta u_{t}}_2^4$. \r{Again, the part $\Sigma_{m_1}$ independent of $\Delta u_t$ in the measurement noise represents the effect of an uncontrolled random disturbance change $\Delta d_{t}$, while the other parts $\Sigma_{m_2}$ and $\Sigma_{m_3}$ encapsulate the second-order error of the Taylor approximation \eqref{eq:taylormeaseq}.} 


To update the sensitivity estimate $\hat{h}_t$, we combine the information given by the previous sensitivity estimate $\hat{h}_{t-1}=\text{vec}(\hat{H}_{t-1})$, and the measurements $\Delta y_{t-1}$ \eqref{eq:measeq}. We compute the new sensitivity estimate $\hat{h}_t$ through a Bayesian update represented in the following least-squares problem \cite{lennart1999system,isermann2010identification}:
    
\begin{align*}
    \begin{split}
            &\hat{h}_{t} = \arg\min_{\hat{h}} \normsz{\hat{h}-\hat{h}_{t-1}}_{{\Sigma_{t-1}^{-1}}}^2 + \normsz{\Delta y_{t-1} - U_{\Delta,t-1} \hat{h}}_{{\Sigma_{m,t-1}^{-1}}}^2,
    \end{split}
\end{align*}  
where $\Sigma_t$ is the covariance matrix representing the uncertainty of the sensitivity estimate $\hat{h}_t$, and $\normsz{x}_A^2=x^TAx$ is the norm of $x$ with respect to a positive definite matrix $A$. The resulting recursive estimation can be expressed as a Kalman filter \cite{jazwinski1970stochfilt}:
\begin{align}\label{eq:sensi}
        \begin{split}
            \hat{h}_{t} = & \hat{h}_{t-1} + K_{t-1} (\Delta y_{t-1} - U_{\Delta,t-1} \hat{h}_{t-1}) \\
            \Sigma_{t} = & \big(\mathbbm{1} - K_{t-1} U_{\Delta,t-1}\big)\Sigma_{t-1} + \Sigma_{p,t-1}, 
        \end{split}
\end{align}
where $\mathbbm{1}$ is the identity matrix, and $K_t = \Sigma_t U_{\Delta,t}^T(\Sigma_{m,t} + U_{\Delta,t}\Sigma_{t} U_{\Delta,t}^T)^{-1}$ is the Kalman gain, \r{which is well defined for an invertible $\Sigma_{m,t}$, see later Assumption~\ref{ass:sensimodel}.}

\begin{rem}
Note that for a diagonal measurement noise covariance $\Sigma_{m,t}=\sigma_{m,t} \mathbbm{1}$, in the limit $\sigma_{m,t} \to \infty$, the gain is $K_t = 0$, thus the sensitivity is not updated, and we keep the initial sensitivity, i.e., $\hat{h}_{t} = \hat{h}_{t-1} = \cdots = \hat{h}_0$. Similarly, a large $\Sigma_{m,t}$ diminishes $K_t$, and helps to tune how fast we want to learn or differ from the initial sensitivity. \r{On the other hand, the process noise covariance $\Sigma_{p,t}$ represents our trust in our current model, and it also helps to tune the learning rate.}
\end{rem}

\subsection{Persistently Exciting OFO}

To learn the time-varying sensitivity $H_t$, we need to capture enough information via the measurement equation \eqref{eq:measeq}, i.e, we need to use different $\Delta u$ to explore different reactions $\Delta y$ and infer different elements of $H_t$ from them. This can be formalized via the persistency of excitation condition \cite{bai1985persexc}: $\Delta u_t$ is persistently exciting if there exists a time span $T>0$, such that for all $t>0$, the matrix formed by columns $\Delta u_{t+i}$ for $i \in \{0,\dots,T\}$ has full rank, i.e., $\text{rank}(\Delta u_{t}, \dots, \Delta u_{t+T}) = n_u$.
To achieve persistency of excitation, we perturb the OFO step \eqref{eq:ofo} with $\omega_{u,t} \in \mathbb{R}^{n_u}$, \r{a bounded zero-mean white noise with independent and identically distributed elements with standard deviation $\sigma_{u}$, e.g., a truncated Gaussian distribution.} As a result, we obtain the following persistently exciting OFO with estimated sensitivity $\hat{H}_t$:
\begin{align}\label{eq:ofope}
    \begin{split}
        u_{t+1} = \Pi_{\mathcal{U}_t}\big[u_t - \alpha \big(\nabla_u f(u_t) + \hat{H}_t^T \nabla_y g(y_t)\big) + \omega_{u,t}\big]
    \end{split}
\end{align}

The resulting interconnected OFO, sensitivity learning and power grid is represented in the block diagram in Figure~\ref{fig:block_diagram}. At each time $t$, a complete loop of the online optimization with sensitivity estimation can be represented as:

\begin{algorithm}[H]\label{alg:alg}
\caption{Online Feedback Optimization (OFO) with sensitivity estimation (blue block in Figure~\ref{fig:block_diagram})}\label{alg:optest}
\begin{algorithmic}[1]
\STATE \textbf{Input:} $y_{t}$ (measured from the grid)
\STATE Recover from previous step: $y_{t-1},u_{t-1},u_t$
\STATE Sensitivity update using \eqref{eq:sensi}: \newline 
$K_{t-1} = \Sigma_{t-1} U_{\Delta,t-1}^T(\Sigma_{m,t-1} + U_{\Delta,t-1}\Sigma_{t-1} U_{\Delta,t-1}^T)^{-1}$ 
$\hat{h}_{t} = \hat{h}_{t-1} + K_{t-1} (\Delta y_{t-1} - U_{\Delta,t-1} \hat{h}_{t-1})$ \newline 
$\Sigma_{t} = \big(\mathbbm{1} - K_{t-1} U_{\Delta,t-1}\big)\Sigma_{t-1} + \Sigma_{p,t-1}$
\STATE Sample the excitation noise $\omega_{u,t} \sim \mathcal{N}(0,\sigma_{u}^2\mathbbm{1})$
\STATE Input optimization using \eqref{eq:ofope}: \newline
$u_{t+1} = \Pi_{\mathcal{U}_t}\big[u_t - \alpha \big(\nabla_u f(u_t) + \hat{H}_t^T \nabla_y g(y_t)\big) + \omega_{u,t}\big]$
\STATE \textbf{Output:} $u_{t+1}$ 
\end{algorithmic}
\end{algorithm}

\begin{rem}\label{rem:otheroper}
The sensitivity learning approach \eqref{eq:sensi} is independent of the method used to update the input $u$, since it only requires the increment $\Delta u$ and the measured $\Delta y$. Hence, it is not only compatible with the projected-gradient-based OFO in \eqref{eq:ofope}, but can be combined with linearly simplified AC-OPF as \eqref{eq:linOPF}, or other OFO approaches, e.g., primal-dual methods \cite{anese2016optimal,ortmann2020experimental}, quadratic programming \cite{haberle2020non,picallo2021qp}, which may have other desirable properties, like strict constraint satisfaction or a faster convergence.
\end{rem}

\subsection{Convergence Analysis}

In this section we analyze the convergence of the estimated sensitivity $\hat{H}_t$ produced by the sensitivity learning \eqref{eq:sensi}, and the input $u_t$ produced by the OFO \eqref{eq:ofope}, towards the true sensitivity $H_t$ and the solution $u^*_t$ of the AC-OPF \eqref{eq:OPF}, respectively. We certify this convergence assuming that the true sensitivity $H_t$ behaves according to the simplified dynamic process \eqref{eq:Hdyn} and satisfies the linear measurements equation \eqref{eq:measeq}; and that the projected gradient descent used in \eqref{eq:ofope} is a strongly monotone and Lipschitz continuous operator:

\begin{defi}[Monotone and Lipschitz operator]\label{ass:monotoneLipschitz}
An operator $F:\mathbb{R}^n \to \mathbb{R}^n$ is $\eta_F$-strongly monotone if $(x_1-x_2)^T(F(x_1)-F(x_2)) \geq \eta_F \normsz{x_1-x_2}_2^2$ for all $x_1,x_2$, and $L_F$-Lipschitz continuous if $\normsz{F(x_1)-F(x_2)}_2 \leq L_F \normsz{x_1-x_2}_2$.
\end{defi}

\begin{ass}\label{ass:sensimodel}
The functions $f(\cdot)$ and $g(\cdot)$ in \eqref{eq:OPF} are continuously differentiable. The sensitivity satisfies \eqref{eq:Hdyn} and \eqref{eq:measeq} with independent $\omega_{p,t}$ and $\omega_{m,t}$. Furthermore, for all $t>0$, $\Sigma_{p,t},\Sigma_{m,t}$ have a positive lower and upper bound, \r{i.e., there exists $\gamma,\beta>0$ such that $\gamma\mathbbm{1} \preceq \Sigma_{p,t} \preceq \beta\mathbbm{1}$, $\gamma\mathbbm{1} \preceq \Sigma_{m,t} \preceq \beta\mathbbm{1}$}; there exists $L_h>0$ such that $\normsz{\nabla_y g(\mathdutchcal{h}(u^*_t,d_t))}_2 \leq L_h$; and the operator $F_t(\cdot)=\nabla_u f(\cdot) + H_t^T \nabla_y g(\mathdutchcal{h}(\cdot,d_t))$ in \eqref{eq:ofope} is $\eta$-strongly monotone and $L$-Lipschitz continuous.
\end{ass}

\r{The continuous differentiability of $f(\cdot)$ and $g(\cdot)$ is common for typical cost functions in power systems, e.g., linear or quadratic $f(\cdot)$, and quadratic penalty functions like $g(\cdot)=\max(0,\cdot)^2$. For strongly convex and Lipschitz smooth cost functions $f(\cdot)$, the strong monotonicity and Lipschitz continuity of the gradient operator $F_t(\cdot)$ holds in certain regions around nominal operating points \cite{colombino2019towards}. \rr{In particular, it would hold if using a usual linear approximation for the input-output map \eqref{eq:ioPF} \cite{picallo2020seopfrt}.} Since $u$ and $d$ are restricted by the grid physical limits, e.g., power ratings, the upper bound of $\normsz{\Delta u_t}_2$ and $\normsz{\nabla_y g(\mathdutchcal{h}(u^*_t,d_t))}_2$ are justified, \rr{since $g(\cdot)$ is differentiable in a compact set}. The persistency of excitation ensures that $\normsz{\Delta u_t}_2>0$ with high probability. Then, \rr{$\Sigma_{p,t}=\Sigma_{p_1}+\Sigma_{p_2}\normsz{\Delta u_{t}}_2^2\succ 0 ,\Sigma_{m,t}=\Sigma_{m_1}+\Sigma_{m_2}\normsz{\Delta u_{t}}_2^2+\Sigma_{m_3}\normsz{\Delta u_{t}}_2^4 \succ 0$} if at least one $\Sigma_{p_i}\succ 0$ and one $\Sigma_{m_j} \succ 0$ for some $i,j$. Finally, even though the true sensitivity is state dependent, i.e., $H_t = \nabla_u \mathdutchcal{h}(u_t,d_t)$, the process and measurement noises in \eqref{eq:Hdyn} and \eqref{eq:measeq} allow to overapproximate the actual behavior of the sensitivity via these simplifications. 
In conclusion, Assumption~\ref{ass:sensimodel} is reasonable.} Then, with a persistently exciting $\Delta u$ as in \eqref{eq:ofope}, we have the following convergence result:

\begin{prop}\label{prop:conv}
Under Assumption~\ref{ass:sensimodel}, and the persistently excited OFO updates \eqref{eq:ofope}, the sensitivity estimates \eqref{eq:sensi} satisfy:
\begin{align}\label{eq:convhproof}
    \begin{split}
        \text{Unbiased mean: } & \normsz{\mathbb{E}[h_t-\hat{h}_t]}_2^2 \leq C_{h,1} e^{-C_{h,2}t}  \overset{t \to \infty}{\to} 0 \\
        \text{Bounded covariance: } & \mathbb{E}[\normsz{h_t-\hat{h}_t}_2^2] = \text{tr}(\Sigma_t) \\
        & \leq C_{h,3} + C_{h,4} e^{-C_{h,5}t} {\to} C_{h,3} ,
    \end{split}
\end{align}
where $\mathbb{E}[\cdot]$ denotes the expectation, $C_{h,i}>0$ are positive constants, and $\overset{t \to \infty}{\to}$ the limit as $t$ goes to infinity. Furthermore, if the step size in \eqref{eq:ofope} satisfies $\alpha < \tfrac{2\eta}{L^2}$, so that $\epsilon = \sqrt{1-2\eta\alpha + L^2\alpha^2} < 1$, then we have

\begin{align}\label{eq:inputconv}
    \begin{split}
        & \mathbb{E}[\normsz{u_{t}-u_{t}^*}_2] \\
        \leq  & \tfrac{1}{1-\epsilon} \big( \sigma_u + \sup_{k < t} \mathbb{E}[\normsz{\Delta u^*_{k}}_2] + \sqrt{C_{h,3}} \alpha L_h \big) \\[-0.1cm]
        & + \hspace{-0.05cm} \epsilon^{t}\mathbb{E}[\normsz{u_0 \hspace{-0.05cm} - \hspace{-0.05cm} u^*_0}_2] \hspace{-0.05cm} + \hspace{-0.05cm} \alpha L_h t \sqrt{C_{h,4}}\max(\epsilon,e^{\frac{-C_{h,5}}{2}})^{t-1} \\[0.1cm]
        \overset{t \to \infty}{\to} & \tfrac{1}{1-\epsilon} \big( \sigma_u + \sup_{k } \mathbb{E}[\normsz{\Delta u^*_{k}}_2] + \sqrt{C_{h,3}}\alpha L_h  \big).
    \end{split}
\end{align}

\end{prop}

\begin{proof}
See Appendix.
\end{proof}

\r{Proposition~\ref{prop:conv} establishes first that the estimated sensitivity $\hat{h}_t$ converges in expectation to the true sensitivity $h_t$ with a bounded covariance. Additionally, the control input $u_t$ converges to the AC-OPF solution $u_t^*$ from \eqref{eq:OPF} with a quantifiable tracking error determined by the bound $C_{h,3}$ of the sensitivity estimation covariance, the variance $\sigma_u$ of the persistency of excitation noise $\omega_u$, and the temporal variation of the AC-OPF solution $\mathbb{E}[\normsz{\Delta u^*_{t}}_2^2]$, where $\Delta u^*_{t}$ can also be bounded by the temporal variation of $d_t$ and $\mathcal{U}_t$ in the AC-OPF \eqref{eq:OPF} \cite{subotic2021quantitative}.}

\section{Test Case}\label{sec:test}

In this section we validate the proposed OFO with sensitivity estimation. We simulate a benchmark distribution grid under time-varying conditions during a 1-hour simulation with 1-second resolution, \r{hence a 1 second control-loop rate. In particular, we show its superior performance against an OFO approach with a constant sensitivity.} First we explain the simulation setup, and then we comment the results obtained.

\subsection{Simulation Setup}\label{subsec:simDef}

\begin{figure}[t]
\centering
\includegraphics[width=8cm,height=6.4cm]{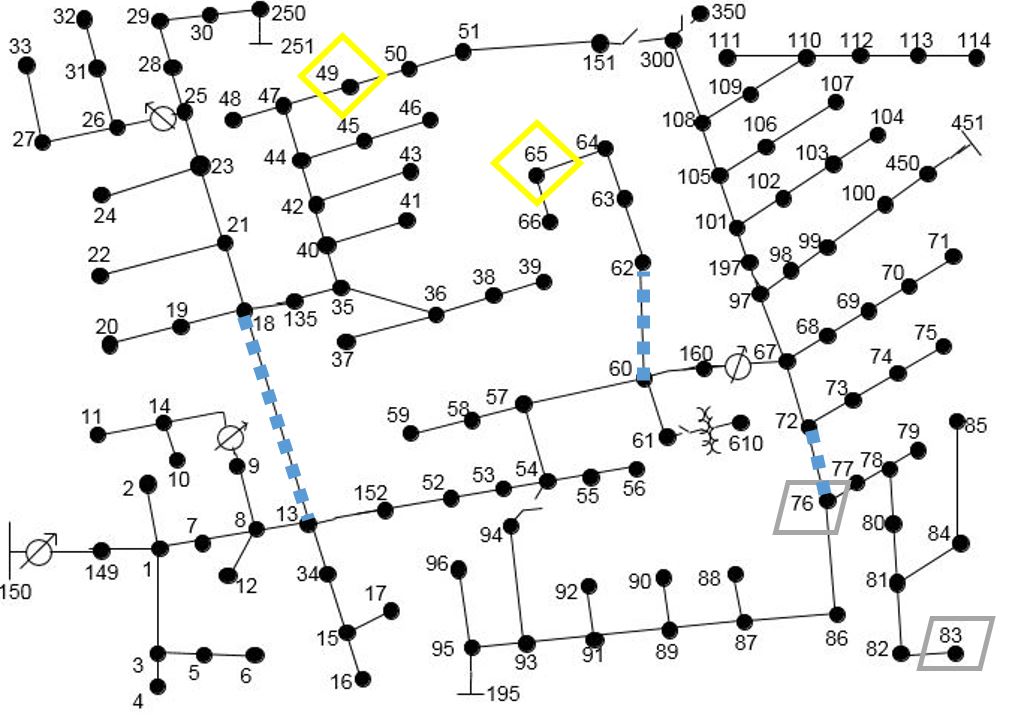}    
\caption{IEEE 123-bus test feeder \cite{kersting1991radial}. \textbf{Distributed generation:} yellow diamond~=~solar, grey parallelogram~=~wind. \textbf{Lines with perturbed electrical parameters:} blue square-dotted.} 
\label{fig:123bus}
\end{figure}
\begin{itemize}[leftmargin=*]
    \item Distribution grid: We use the 3-phase, unbalanced IEEE 123-bus test feeder \cite{kersting1991radial} in Figure~\ref{fig:123bus}. 
    
    \item Disturbance $d$: We consider uncontrollable active and reactive loads in our disturbance vector $d$. To generate these load profiles we use $1$-second resolution data of the ECO data set \cite{ECOdata}, then aggregate households and rescale them to the base loads of the 123-bus feeder. This gives us values of $d_t$ for every second during simulation time of 1h. 
    
    \item Controllable input set-points $u$: We add two solar PV systems and two wind turbines to the grid 
    as in \cite{picallo2020seopfrt}, see Figure~\ref{fig:123bus}. They can inject active power, and inject and absorb reactive power on all three phases, which gives us 24 control inputs. We consider a slack bus 150 in Figure~\ref{fig:123bus}, with a controllable voltage magnitude through, e.g., a tap changer, which makes in total $n_u=25$. The solar and wind generation profiles are generated based on a $1$-minute solar irradiation profile \cite{solarprofile} and a $2$-minute wind speed profile \cite{windprofile}. Generation is assumed constant between samples. We use these profiles to set the time-varying upper limit of the feasible set $\overline{u}_{t}$, set the lower limit of active generation to $\underline{u}_t=0$, and define $\mathcal{U}_t = \{u\,|\,\underline{u}_t \leq u \leq \overline{u}_t\}$. 
    
    \item Output $y$: We consider as output $y$ the voltage magnitudes of all phases at all buses except the slack bus, given that it is a control input.
    
    \item AC-OPF cost function in \eqref{eq:OPF}: We use a quadratic cost that penalizes deviating from a reference: $f(u)=\frac{1}{2}\normsz{u-u_\text{ref}}_2^2$. The reference $u_\text{ref}$ for the voltage magnitude at the slack bus is $1$~p.u. The reference for the controllable generation is the maximum installed power to promote using as much renewable energy as possible. The reference for reactive power is $0$. Note that the cost function is continuously differentiable, and has a strongly monotone and Lipschitz continuous gradient as required in Assumption~\ref{ass:sensimodel}. We consider the voltage limits $[0.94\, \text{p.u.},1.06\,\text{p.u.}]$ for all nodes as in \cite{hauswirth2017online,picallo2020seopfrt}, and use the penalty function $g(y)= \frac{\rho}{2}\max\big( \left[\begin{smallmatrix} \mathbbm{1}  \\ -\mathbbm{1}  \end{smallmatrix}\right] y + \left[\begin{smallmatrix} -1.06 \\ 0.94 \end{smallmatrix}\right],0 \big)^2$, with a sufficiently large penalization parameter $\rho=100$ to discourage violations. Again, this function is continuously differentiable, and has a monotone and Lipschitz continuous gradient.
    
    \item \r{Sensitivity process and measurement noises in \eqref{eq:Hdyn} and \eqref{eq:measeq}: Under fast sampling rates $\Delta d_t$ may be negligible, especially when compared to $\Delta u_t$. Hence, for the simulation we assign $\Sigma_{p_1},\Sigma_{m_1},\Sigma_{m_2}$ to 0, and keep $\Sigma_{p_2},\Sigma_{m_3} \succ 0$. This ensures that $\Sigma_{p,t},\Sigma_{m,t}\succ 0$ for all $t$, as required by Assumption~\ref{ass:sensimodel}.}
    
    \item Persistency of excitation: We use a symmetric truncated Gaussian distribution with $\sigma_u=0.0001$~p.u. to introduce a low persistency of excitation noise $\omega_{u,t}$ that facilitates our sensitivity learning, but avoids introducing a big deviation in the input convergence, see \eqref{eq:inputconv}. 
    
    \item Initializing sensitivity and linear model \eqref{eq:linPF}: We use the zero-injection operating point $u_\text{op}=0,d_\text{op}=0$ to initialize the sensitivity estimation, i.e., $\hat{H}_0= H_0=\nabla_u \mathdutchcal{h}(u,d)|_{(0,0)}$, see \eqref{eq:linPF}. In the first simulation (1: true admittance) we use the true admittance to compute $H_0$, in the second (2: perturbed admittance) we use a perturbed admittance matrix, where we have introduce an up to $20\%$ error in the admittance of the lines indicated in Figure~\ref{fig:123bus}.

\end{itemize}

\subsection{Results}
We analyze the simulation performance of OFO with sensitivity learning \eqref{eq:sensi} and \eqref{eq:ofope}, and compare it against an OFO with constant sensitivity \eqref{eq:ofo}. We validate both results in Proposition~\ref{prop:conv}: First, the estimated sensitivity $\hat{H}_t$ converges to the real time-varying sensitivity $H_t$. Second, the input $u_t$ converges to the AC-OPF solution $u^*_t$ \eqref{eq:OPF}. 

\subsubsection{True admittance}
First we perform a simulation where we use the true admittance to derive the initial sensitivity $H_0$ in the linear power flow approximation \eqref{eq:linPF}. Figure~\ref{fig:opterror} shows the norm of the AC-OPF solution $u^*_t$ of \eqref{eq:OPF} that we calculate with the correct non-linear model $\mathdutchcal{h}(\cdot)$ and the disturbances $d_t$. This optimal input is time-varying due to the changing solar radiation and wind speed in the limits $\overline{u}_{t}$, and the temporal variation of the loads in $d_t$.  Figure~\ref{fig:opterror} shows how the OFO control input $u_t$ converges towards the optimal input $u^*_t$ using different sensitivities: The inputs $u_{H}$ produced by the OFO controller \eqref{eq:ofo} with the exact sensitivity $H_t=\nabla \mathdutchcal{h}(u_t,d_t)$ succeed in tracking the AC-OPF solution $u^*$, with relatively small differences caused by the time-varying disturbances $d_t$ and/or available energy $\overline{u}_t$. However, when using the constant sensitivity $H_0$ in \eqref{eq:ofo}, there is a large difference between the generated control input $u_{H_0}$ and the optimal one $u^*$. This gap is closed when using the OFO with sensitivity estimation \eqref{eq:ofope}, i.e., $u_{\hat{H}}$ is able to converge to the AC-OPF solution $u^*$ of \eqref{eq:OPF} with a small tracking error, as predicted by Proposition~\ref{prop:conv}.

\r{Figure~\ref{fig:sensierror} shows the relative error $\frac{\normsz{\Delta y - H\Delta u}_2}{\normsz{\Delta y}_2}$ of the measurement equation \eqref{eq:measeq}. This helps to understand why OFO with sensitivity learning \eqref{eq:sensi} performs better than with a constant sensitivity $H_0$: The linearization error with estimated sensitivity $\hat{H}_t$ gets lower respect to the one with $H_0$. This means that the learned sensitivity becomes a more accurate linear approximation than \eqref{eq:linPF}, which causes the lower optimization error observed in Figure~\ref{fig:opterror}. Even though the error $\frac{\normsz{\Delta y - H\Delta u}_2}{\normsz{\Delta y}_2}$ does not converge to $0$ when using $\hat{H}$, the sensitivity estimation approach \eqref{eq:sensi} learns enough to drive the control set-points to the optimum, see Figure~\ref{fig:opterror}, which is our ultimate objective.}

\r{Finally, Figure~\ref{fig:Vlimits} shows that the inputs $u_{\hat{H}}$, produced by the OFO with sensitivity estimation \eqref{eq:ofope} result into much less voltage violations than $u_{H_0}$ from the OFO with constant sensitivity \eqref{eq:ofo}. Actually, the number of voltage violations of $u_{\hat{H}}$ gets close to those of the OFO with true sensitivity $u_H$. Hence, the OFO with sensitivity estimation not only reduces the distance to the AC-OPF solution, see Figure~\ref{fig:opterror}, but performs a better voltage regulation.}

\subsubsection{Perturbed admittance}
In Figure~\ref{fig:opterrornoisy} we show a simulation for which we perturb the admittance of the lines indicated in Figure~\ref{fig:123bus} with an up to $20\%$ error. We observe how the OFO with sensitivity learning $u_{\hat{H}}$ \eqref{eq:ofope} is still able to track the AC-OPF solution $u^*$ of \eqref{eq:OPF} in time-varying conditions. The OFO $u_{H_0}$ with a fixed sensitivity \eqref{eq:ofo} and the same step size as $u_{\hat{H}}$ diverges, since it tries to regulate the voltage with a wrong sensitivity that is too far from the actual one. Convergence is recovered with a lower step size in $u_{H_0,\text{slow}}$, but it still performs poorly at tracking the AC-OPF solution. \r{This experiment allows us to conclude that the OFO with sensitivity estimation \eqref{eq:ofope} is a model-free approach that does not require an accurate model, but learns it online. 
}


\begin{figure}[t]
\centering
\includegraphics[width=9cm]{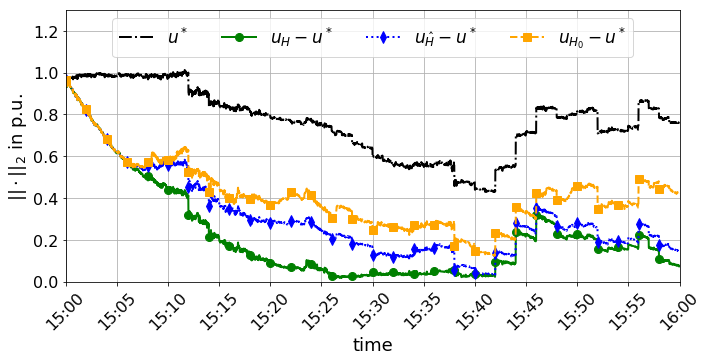}    
\caption{Euclidean norm of the AC-OPF solution $u_{t}^*$, and the optimization error between $u_{t}^*$ and the set-points $u_t$ produced by the OFO \eqref{eq:ofo}, using either the true sensitivity $H$ (green with dots), the estimated sensitivity ${\hat{H}}$ (blue with diamonds), the constant sensitivity at a zero-injection operation point ${H_0}$ (yellow with squares), with respective set-points $u_H,u_{\hat{H}},u_{H_0}$.} 
\label{fig:opterror}
\end{figure}

\begin{figure}[t]
\centering
\includegraphics[width=9cm]{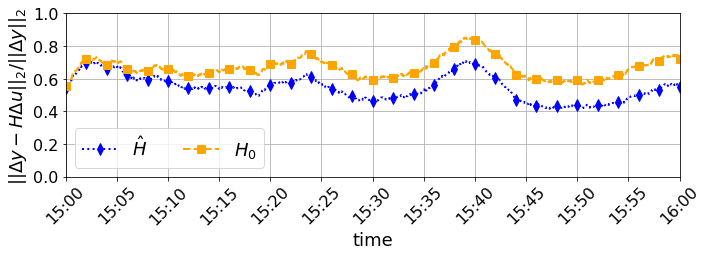}    
\caption{Moving average over 5 minutes of the relative error $\frac{\normsz{\Delta y - H\Delta u}_2}{\normsz{\Delta y}_2}$ when using the learned sensitivity $\hat{H}$ (blue with diamonds) or the one fixed at an zero-injection operation point $H_0$ (yellow with squares).}  
\label{fig:sensierror}
\end{figure}


\begin{figure}[t]
\centering
\includegraphics[width=9cm]{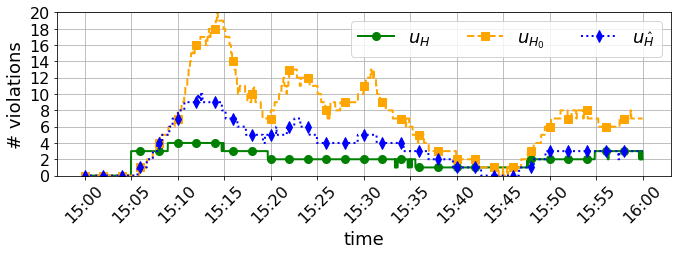}    
\caption{Moving average over 5 minutes of the number of voltage violations across all nodes.} 
\label{fig:Vlimits}
\end{figure}

\begin{figure}[t]
\centering
\includegraphics[width=9cm]{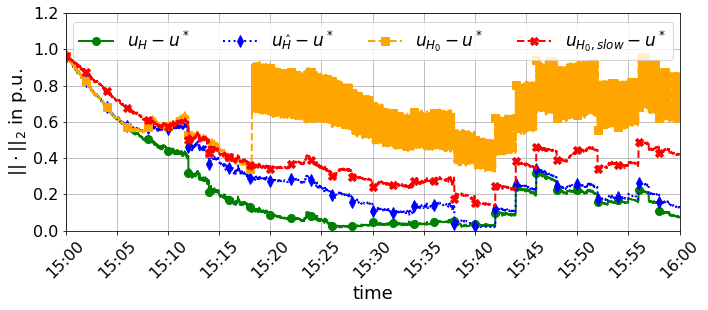}    
\caption{
Same as Figure~\ref{fig:opterror}. For the constant sensitivity $H_0$, we plot $u_{H_0}$ (yellow) when using the same step size as $u_H,u_{\hat{H}}$, and $u_{H_0,\text{slow}}$ (red) with a smaller step size. Both $\hat{H}$ and $H_0$ are initialized with an perturbed admittance matrix $Y$.} 
\label{fig:opterrornoisy}
\end{figure}

\section{Conclusion and Outlook}\label{sec:conc}
Standard Online Feedback Optimization (OFO) typically uses an approximate input-output sensitivity, which may lower its performance. Alternative, one can compute the actual sensitivity, but that requires, having an accurate grid model and full grid observability, which is usually not available.
In this work we have proposed a recursive estimation approach that provides Online Feedback Optimization (OFO) with a tool to learn the model sensitivity without extensive measurements, and thus improves its performance and turns OFO into a model-free approach. We have provided convergence guarantees when approximating the time-varying sensitivity behavior by a random process with linear measurements. We have established that even under time-varying conditions the estimated sensitivity and the control input converge to a neighborhood of the true sensitivity and the solution of the AC-OPF, respectively. \r{Finally, we have validated with simulations using the IEEE 123-bus test feeder that our proposed OFO controller with sensitivity estimation performs successfully even though the actual sensitivity is state-dependent}, i.e., it is able to track a time-varying optimal input while satisfying the grid specifications. In short, the proposed OFO controller with sensitivity estimation can be used as a model-free plug-and-play controller for real-time power grid operation that enables safe and optimal control. 

An interesting future addition would be to investigate a more suitable way to design the persistency of excitation, possibly linked to the optimization problem, so that it explores specific directions of interest. 
Additionally, it would be interesting to observe how the proposed sensitivity estimation approach performs under a sudden change of topology caused by, e.g., a line fault, network split, etc.; \rr{under communication problems, e.g., delays,  missing packages, recurrent outliers due to, for example, sensor misscalibration.} 

\bibliographystyle{IEEEtran}
\bibliography{IEEEabrv,biblio}

\begin{thebibliography}{10}
\providecommand{\url}[1]{#1}
\csname url@samestyle\endcsname
\providecommand{\newblock}{\relax}
\providecommand{\bibinfo}[2]{#2}
\providecommand{\BIBentrySTDinterwordspacing}{\spaceskip=0pt\relax}
\providecommand{\BIBentryALTinterwordstretchfactor}{4}
\providecommand{\BIBentryALTinterwordspacing}{\spaceskip=\fontdimen2\font plus
\BIBentryALTinterwordstretchfactor\fontdimen3\font minus
  \fontdimen4\font\relax}
\providecommand{\BIBforeignlanguage}[2]{{%
\expandafter\ifx\csname l@#1\endcsname\relax
\typeout{** WARNING: IEEEtran.bst: No hyphenation pattern has been}%
\typeout{** loaded for the language `#1'. Using the pattern for}%
\typeout{** the default language instead.}%
\else
\language=\csname l@#1\endcsname
\fi
#2}}
\providecommand{\BIBdecl}{\relax}
\BIBdecl

\bibitem{molzahn2019surveyrel}
D.~K. Molzahn and I.~A. Hiskens, ``{A Survey of Relaxations and Approximations
  of the Power Flow Equations},'' \emph{Foundations and Trends in Electric
  Energy Systems}, vol.~4, no. 1-2, pp. 1--221, February 2019.

\bibitem{bolognani2013topol}
S.~Bolognani, N.~Bof, D.~Michelotti, R.~Muraro, and L.~Schenato,
  ``Identification of power distribution network topology via voltage
  correlation analysis,'' in \emph{52nd Conf. on Decision and Control}, 2013,
  pp. 1659--1664.

\bibitem{moffat2020pmuimpedance}
K.~Moffat, M.~Bariya, and A.~Von~Meier, ``Unsupervised impedance and topology
  estimation of distribution networks—limitations and tools,'' \emph{{IEEE}
  Trans. Smart Grid}, vol.~11, no.~1, pp. 846--856, 2020.

\bibitem{molzahn2017survey}
D.~K. {Molzahn}, F.~{D{\"o}rfler}, H.~{Sandberg}, S.~H. {Low},
  S.~{Chakrabarti}, R.~{Baldick}, and J.~{Lavaei}, ``A survey of distributed
  optimization and control algorithms for electric power systems,''
  \emph{{IEEE} Trans. Smart Grid}, vol.~8, no.~6, pp. 2941--2962, Nov. 2017.

\bibitem{hauswirth2017online}
A.~Hauswirth, A.~Zanardi, S.~Bolognani, F.~D{\"o}rfler, and G.~Hug, ``Online
  optimization in closed loop on the power flow manifold,'' in \emph{2017 IEEE
  Manchester PowerTech}.\hskip 1em plus 0.5em minus 0.4em\relax IEEE, 2017, pp.
  1--6.

\bibitem{anese2016optimal}
E.~Dall'Anese and A.~Simonetto, ``Optimal power flow pursuit,'' \emph{{IEEE}
  Trans. Smart Grid}, vol.~9, no.~2, pp. 942--952, Mar. 2018.

\bibitem{ortmann2020experimental}
L.~Ortmann, A.~Hauswirth, I.~Caduff, F.~D{\"o}rfler, and S.~Bolognani,
  ``Experimental validation of feedback optimization in power distribution
  grids,'' \emph{Electric Power Systems Research}, vol. 189, p. 106782, 2020.

\bibitem{picallo2020seopfrt}
M.~Picallo, S.~Bolognani, and F.~Dörfler, ``Closing the loop: Dynamic state
  estimation and feedback optimization of power grids,'' \emph{Electric Power
  Systems Research}, vol. 189, p. 106753, 2020.

\bibitem{bolognani2015fast}
S.~Bolognani and F.~D{\"o}rfler, ``Fast power system analysis via implicit
  linearization of the power flow manifold,'' in \emph{53rd Annual Allerton
  Conf. on Communication, Control, and Computing}.\hskip 1em plus 0.5em minus
  0.4em\relax IEEE, 2015, pp. 402--409.

\bibitem{colombino2019towards}
M.~Colombino, J.~W. Simpson-Porco, and A.~Bernstein, ``Towards robustness
  guarantees for feedback-based optimization,'' in \emph{2019 IEEE 58th Conf.
  on Decision and Control}.\hskip 1em plus 0.5em minus 0.4em\relax IEEE, 2019,
  pp. 6207--6214.

\bibitem{chen2021reinforcement}
X.~Chen, G.~Qu, Y.~Tang, S.~Low, and N.~Li, ``Reinforcement learning for
  decision-making and control in power systems: Tutorial, review, and vision,''
  \emph{arXiv preprint arXiv:2102.01168}, 2021.

\bibitem{coulson2019data}
J.~Coulson, J.~Lygeros, and F.~D{\"o}rfler, ``Data-enabled predictive control:
  In the shallows of the deepc,'' in \emph{2019 18th European Control
  Conference (ECC)}.\hskip 1em plus 0.5em minus 0.4em\relax IEEE, 2019, pp.
  307--312.

\bibitem{mugnier2016model}
C.~Mugnier, K.~Christakou, J.~Jaton, M.~De~Vivo, M.~Carpita, and M.~Paolone,
  ``Model-less/measurement-based computation of voltage sensitivities in
  unbalanced electrical distribution networks,'' in \emph{2016 Power Systems
  Computation Conference (PSCC)}.\hskip 1em plus 0.5em minus 0.4em\relax IEEE,
  2016, pp. 1--7.

\bibitem{bianchin2021data}
G.~Bianchin, M.~Vaquero, J.~Cortes, and E.~Dall'Anese, ``Data-driven synthesis
  of optimization-based controllers for regulation of unknown linear systems,''
  Dec. 2021.

\bibitem{willems2005note}
J.~C. Willems, P.~Rapisarda, I.~Markovsky, and B.~L. De~Moor, ``A note on
  persistency of excitation,'' \emph{Systems \& Control Letters}, vol.~54,
  no.~4, pp. 325--329, 2005.

\bibitem{chen2021model}
X.~Chen, J.~I. Poveda, and N.~Li, ``Model-free optimal voltage control via
  continuous-time zeroth-order methods,'' Dec. 2021.

\bibitem{lennart1999system}
L.~Lennart, ``System identification: theory for the user,'' \emph{PTR Prentice
  Hall, Upper Saddle River, NJ}, vol.~28, 1999.

\bibitem{isermann2010identification}
R.~Isermann and M.~M{\"u}nchhof, \emph{Identification of dynamic systems: an
  introduction with applications}.\hskip 1em plus 0.5em minus 0.4em\relax
  Springer Science \& Business Media, 2010.

\bibitem{kersting1991radial}
W.~Kersting, ``Radial distribution test feeders,'' \emph{{IEEE} Trans. Power
  Syst.}, vol.~6, no.~3, pp. 975--985, 1991.

\bibitem{krantz2012implicit}
S.~G. Krantz and H.~R. Parks, \emph{The implicit function theorem: history,
  theory, and applications}.\hskip 1em plus 0.5em minus 0.4em\relax Springer
  Science \& Business Media, 2012.

\bibitem{low2014convex}
S.~H. Low, ``Convex relaxation of optimal power flow—part i: Formulations and
  equivalence,'' \emph{IEEE Transactions on Control of Network Systems},
  vol.~1, no.~1, pp. 15--27, 2014.

\bibitem{bolognani2015existence}
S.~Bolognani and S.~Zampieri, ``On the existence and linear approximation of
  the power flow solution in power distribution networks,'' \emph{IEEE
  Transactions on Power Systems}, vol.~31, no.~1, pp. 163--172, 2015.

\bibitem{bertsekas1997nonlinear}
D.~P. Bertsekas, ``Nonlinear programming,'' \emph{Journal of the Operational
  Research Society}, vol.~48, no.~3, pp. 334--334, 1997.

\bibitem{jazwinski1970stochfilt}
A.~H. Jazwinski, ``Mathematics in science and engineering,'' \emph{Stochastic
  processes and filtering theory}, vol.~64, 1970.

\bibitem{bai1985persexc}
E.~Bai and S.~Sastry, ``Persistency of excitation, sufficient richness and
  parameter convergence in discrete time adaptive control,'' \emph{Systems \&
  Control Letters}, vol.~6, no.~3, pp. 153--163, 1985.

\bibitem{haberle2020non}
V.~H{\"a}berle, A.~Hauswirth, L.~Ortmann, S.~Bolognani, and F.~D{\"o}rfler,
  ``Non-convex feedback optimization with input and output constraints,''
  \emph{IEEE Control Systems Letters}, vol.~5, no.~1, pp. 343--348, 2020.

\bibitem{picallo2021qp}
M.~Picallo, D.~Liao-McPherson, S.~Bolognani, and F.~D{\"o}rfler, ``Cross-layer
  design for real-time grid operation: Estimation, optimization and power
  flow,'' \emph{arXiv preprint arXiv:2109.13842}, 2021.

\bibitem{subotic2021quantitative}
I.~Subotic, A.~Hauswirth, and F.~Dorfler, ``Quantitative sensitivity bounds for
  nonlinear programming and time-varying optimization,'' \emph{IEEE
  Transactions on Automatic Control}, 2021.

\bibitem{ECOdata}
C.~Beckel, W.~Kleiminger, R.~Cicchetti, T.~Staake, and S.~Santini, ``The {ECO}
  data set and the performance of non-intrusive load monitoring algorithms,''
  in \emph{Proc. 1st ACM Conf. on Embedded Systems for Energy-Efficient
  Buildings}, 11 2014.

\bibitem{solarprofile}
HelioClim-3, ``{HelioClim-3 Database of Solar Irradiance},''
  \url{http://www.soda-pro.com/web-services/radiation/helioclim-3-archives-for-free},
  accessed: 2017-12-01.

\bibitem{windprofile}
MERRA-2, ``{The Modern-Era Retrospective analysis for Research and
  Applications, Version 2 (MERRA-2) Web service},''
  \url{http://www.soda-pro.com/web-services/meteo-data/merra}, accessed:
  2017-12-01.

\bibitem{tarn1976observers}
{Tzyh-Jong Tarn} and Y.~{Rasis}, ``Observers for nonlinear stochastic
  systems,'' \emph{{IEEE} Trans. Autom. Control}, vol.~21, no.~4, pp. 441--448,
  Aug. 1976.

\end{thebibliography}


\section*{Appendix: Proof of Proposition~\ref{prop:conv}}\label{app:proof}

Consider the information matrix $W_I = \sum_{k=t}^{t+T} U_{\Delta,k}^T {\Sigma_{m,k}^{-1}} U_{\Delta,k} \allowbreak =\sum_{k=t}^{t+T} (\Delta u_k \Delta u_k^T) \otimes {\Sigma_{m,k}^{-1}}$. 
Since $\gamma \mathbbm{1} \preceq \Sigma_{m,t} \preceq \beta \mathbbm{1}$ for all $t$, we have $\frac{1}{\beta} \mathbbm{1} \preceq \Sigma_{m,t}^{-1} \preceq \frac{1}{\gamma} \mathbbm{1}$, and $(\sum_{k=t}^{t+T} \Delta u_k \Delta u_k^T) \otimes \frac{1}{\beta}\mathbbm{1} \preceq W_I \preceq (\sum_{k=t}^{t+T} \Delta u_k \Delta u_k^T) \otimes \frac{1}{\gamma}\mathbbm{1}$. 
Since $\Delta u$ is persistently exciting, there exists a sufficiently large $T$ and $\gamma_2,\beta_2>0$ so that $\gamma_2\mathbbm{1} \preceq \sum_{k=t}^{t+T} \Delta u_k \Delta u_k^T \preceq \beta_2\mathbbm{1}$, and thus $\frac{\gamma_2}{\beta}\mathbbm{1} \preceq W_I \preceq \frac{\beta_2}{\gamma}\mathbbm{1}$. 
Hence, the matrix pair $(\mathbbm{1},U_{\Delta,t})$ from the dynamic system \eqref{eq:Hdyn} and \eqref{eq:measeq} is uniformly completely observable, and, additionally, uniformly complete controllable given $\Sigma_{p,t}\succ 0$ \cite[Ch.~7]{jazwinski1970stochfilt}. As a result, the sensitivity converges exponentially in expectation, and is exponentially bounded in mean square \cite{jazwinski1970stochfilt, tarn1976observers}, i.e., there exists positive constants $C_{h,i}>0$ satisfying \eqref{eq:convhproof}.

\rr{Then, under Assumption~\ref{ass:monotoneLipschitz} we have
\begin{align*}
    \begin{split}
        & \normsz{u_{t+1}-u_{t+1}^*}_2 \leq \normsz{u_{t+1}-u^*_{t}}_2 + \normsz{\Delta u^*_{t}}_2 
        \\[-0.1cm]
        \overset{\eqref{eq:ofope}}{\hspace{0.4cm}\leq} & \normsz{\Pi_{\mathcal{U}_t}\big[u_t - \alpha \big(\nabla_u f(u_t) + \hat{H}_t^T \nabla_y g(y_t)\big) + \omega_{u,t}\big] \\
        & - \Pi_{\mathcal{U}_t}\big[u^*_t - \alpha F_t(u^*_t) \big]}_2 + \normsz{\Delta u^*_{t}}_2 
        \\
        {\leq} & \normsz{\big(u_t - \alpha \big(\nabla_u f(u_t) \hspace{-0.05cm} + \hspace{-0.05cm} \hat{H}_t^T \nabla_y g(y_t)\big) + \omega_{u,t}\big) \pm H_t \nabla_y g(y_t)\\
        & - \big(u^*_t - \alpha F_t(u^*_t) \big)}_2 + \normsz{\Delta u^*_{t}}_2 
        \\
        \leq & \normsz{\big(u_t - \alpha F_t(u_t)\big) - \big(u^*_t - \alpha F_t(u^*_t) \big)}_2 + \normsz{\omega_{u,t}}_2\\
        & + \alpha L_h\normsz{h_t-\hat{h}_t}_2 + \normsz{\Delta u^*_{t}}_2 
        \\
        \leq & \epsilon \normsz{u_t  - u^*_t}_2 + \normsz{\omega_{u,t}}_2 + \alpha L_h\normsz{h_t -  \hat{h}_t}_2 + \normsz{\Delta u^*_{t}}_2,
    \end{split}
\end{align*}
where in the second inequality we use that $u^*_t$ satisfies $u^*_t = \Pi_{\mathcal{U}_t}\big[u^*_t - \alpha F_t(u_t^*)\big]$, i.e., due to optimality $u^*_t$ is a fixed point of the operator \eqref{eq:ofope} with $\omega_{u,t}=0$ and the true sensitivity $H_t$ instead of the estimated one $\hat{H}_t$. In the fourth inequality, where $\epsilon^2 = 1-2\eta\alpha + L^2\alpha^2$, we use that the operator $F_t(\cdot)$ is $\eta$-strongly monotone and $L$-Lipschitz continuous. Hence, in expectation we have
\begin{align*}
    \begin{split}
        & \mathbb{E}[\normsz{u_{t+1}-u_{t+1}^*}_2] 
        \\
        \leq & \epsilon \mathbb{E}[\normsz{u_t-u^*_t}_2] + \sigma_u + \mathbb{E}[\normsz{\Delta u^*_{t}}_2] + \alpha L_h \mathbb{E}[\normsz{h_t-\hat{h}_t}_2] 
        \\
        \leq & \epsilon^{t+1}\mathbb{E}[\normsz{u_0-u^*_0}_2] + \tfrac{1}{1-\epsilon} \big( \sigma_u + \sup_{k\leq t} \mathbb{E}[\normsz{\Delta u^*_{k}}_2] \big) \\[-0.35cm]
        & + \alpha L_h \sum_{k=0}^{t} \epsilon^{t-k}\mathbb{E}[\normsz{h_k-\hat{h}_k}_2] 
        \\[-0.2cm]
        \overset{\eqref{eq:convhproof}}{\leq} &
        \epsilon^{t+1}\mathbb{E}[\normsz{u_0-u^*_0}_2] \\ 
        & + \tfrac{1}{1-\epsilon} \big( \sigma_u + \sup_{k\leq t} \mathbb{E}[\normsz{\Delta u^*_{k}}_2] + \sqrt{C_{h,3}}\alpha L_h \big) \\[-0.15cm]
        & + \alpha L_h (t+1) \sqrt{C_{h,4}}\max(\epsilon,e^{\frac{-C_{h,5}}{2}})^t
        \\[-0.1cm]
        \overset{t \to \infty}{\to} & \tfrac{1}{1-\epsilon} \big( \sigma_u + \sup_{k} \mathbb{E}[\normsz{\Delta u^*_{k}}_2] + \sqrt{C_{h,3}}\alpha L_h \big),
    \end{split}
\end{align*}
where in the second inequality we apply the first one recursively. In the second and third inequality we bound the geometric series $\sum_{k=0}^t \epsilon^{t-k} \leq \frac{1}{1-\epsilon}$, and use that $\sqrt{\cdot}$ is subadditive.}

\end{document}